\documentclass[a4paper,reqno]{amsart}
\pdfoutput=1

\usepackage{ifthen}
\usepackage{xargs}

\newcommand{\prd}[1]{\Pi_{#1}\mkern1mu}

\newcommand{\sm}[1]{\Sigma \left(#1\right).\,}
\newcommand{\smsimple}[1]{\Sigma_{#1}}

\newcommand{\lam}[1]{\lambda #1 .}

\newcommand{\jdeq}{\equiv}      
\newcommand{\defeq}{\vcentcolon\equiv}  

\newcommand{\idfunc}[1][]{\ensuremath{\mathsf{id}_{#1}}\xspace}

\newcommand{\idtype}[3][]{\ensuremath{\mathsf{Id}_{#1}(#2,#3)}\xspace}

\newcommand{\refl}[1]{\ensuremath{\mathsf{refl}_{#1}}\xspace}

\newcommand{\trans}[2]{\ensuremath{{#1}_{*}\mathopen{}\left({#2}\right)\mathclose{}}\xspace}

\newcommand{\transf}[1]{\ensuremath{{#1}_{*}}\xspace} 

\newcommand{\UU}{\ensuremath{\mathcal{U}}\xspace}

\newcommand{\N}{\ensuremath{\mathbb{N}}\xspace}

\newcommand{\suc}{\mathsf{succ}}

\newcommand{\emptyt}{\ensuremath{\mathbf{0}}\xspace}

\newcommand{\unit}{\ensuremath{\mathbf{1}}\xspace}
\newcommand{\ttt}{\ensuremath{\star}\xspace}

\newcommand{\iscontr}{\ensuremath{\mathsf{isContr}}}
\newcommand{\isprop}{\ensuremath{\mathsf{isProp}}}
\newcommand{\isset}{\ensuremath{\mathsf{isSet}}}

\def\compare#1#2#3#4{\if#1#3\if#2#41\else0\fi\else0\fi}

\newcommand{\istype}[1]{
  \edef\a{\compare-2#1\empty\empty}
  \if\a1 \iscontr \else
  \edef\b{\compare-1#1\empty\empty}
  \if\b1 \isprop \else
  \edef\c{#1}
  \if0\c \isset \else
  \mathsf{is}\mbox{-}{#1}\mbox{-}\mathsf{type} \fi\fi\fi
}

\newcommand{\UIP}{\textsf{UIP}}
\newcommand{\MLTT}{\textsf{MLTT}}
\newcommand{\HOTT}{\textsf{HoTT}}
\newcommand{\MLTTK}{\ensuremath{\MLTT_K}}

\newcommand{\trunc}[2]{\mathopen{}\left\Vert #2\right\Vert_{#1}\mathclose{}}

\newcommand{\ordinal}[1]{[\mathsf{#1}]} 

\newcommand{\fst}{\mathsf{fst}}

\newcommand{\inl}{\mathsf{inl}\xspace}
\newcommand{\inr}{\mathsf{inr}\xspace}

\newcommand{\deltop}{\ensuremath{\Delta_+^\mathrm{op}}}
\newcommand{\deltplus}{\ensuremath{\Delta_+}}
\newcommand{\C}{\mathcal{C}}

\renewcommand{\lim}{\mathsf{lim}\xspace}

\newcommand{\oppo}[1]{#1^\mathrm{op}}

\newcommand{\zero}{\ensuremath{\mathbf{0}}\xspace}

\newcommand{\Fin}{\mathsf{Fin}}
\newcommand{\isIncr}{\mathsf{isIncr}}

\newcommand{\strict}[1]{#1^\mathrm{s}}
\newcommand{\fibrant}[1]{#1^\mathrm{f}}
\newcommand{\Ty}{\mathsf{Ty}}
\newcommand{\fTy}{\fibrant{\Ty}}
\newcommand{\Tm}{\mathsf{Tm}}
\newcommand{\fTm}{\fibrant{\Tm}}
\newcommand{\op}[1]{#1^{\mathrm{op}}}
\newcommand{\setcat}{\mathsf{Set}}
\newcommand{\presheaf}[1]{\widehat{#1}}

\usepackage{microtype}
\usepackage{enumitem}
\usepackage{csquotes}
\usepackage{mathalfa}
\usepackage{amsmath}
\usepackage{amsfonts}
\usepackage[T1]{fontenc}
\usepackage{listings}
\usepackage{mathtools}
\usepackage[llparenthesis,rrparenthesis]{stmaryrd}
\usepackage{xspace}
\usepackage{tabularx}

\usepackage{smalllists}

\usepackage{graphics}
\usepackage{MnSymbol}
\usepackage{braket}
\usepackage{float}
\usepackage{enumitem}
\usepackage{amsthm}
\usepackage{hyperref}
\usepackage[normalem]{ulem}

\usepackage{proof}

\theoremstyle{plain}
\newtheorem{theorem}{Theorem}[section]
\newtheorem{lemma}[theorem]{Lemma}

\theoremstyle{definition}
\newtheorem{remark}[theorem]{Remark}

\newtheorem{definition}[theorem]{Definition}

\title{Extending Homotopy Type Theory with Strict Equality}
\author{Thorsten Altenkirch \and Paolo Capriotti \and Nicolai Kraus}
\thanks{This work was supported by the Engineering and Physical Sciences Research Council (EPSRC), grant reference EP/M016994/1, and by USAF, Airforce office for scientific research, award FA9550-16-1-0029.}

\newcommand{\Type}{\textbf{Type}}

\begin{document}

\begin{abstract}
In homotopy type theory (HoTT), all constructions are necessarily stable under homotopy equivalence.  This has shortcomings: for example, it is believed that it is impossible to define a type of semi-simplicial types.  More generally, it is difficult and often impossible to handle towers of coherences.  To address this, we propose a 2-level theory which features both strict and weak equality.  This can essentially be represented as \emph{two} type theories: an ``outer'' one, containing a strict equality type former, and an ``inner'' one, which is some version of HoTT.  Our type theory is inspired by Voevodsky's suggestion of a \emph{homotopy type system} (HTS) which currently refers to a range of ideas.  A core insight of our proposal is that we do not need any form of equality reflection in order to achieve what HTS was suggested for.  Instead, having unique identity proofs in the outer type theory is sufficient, and it also has the meta-theoretical advantage of not breaking decidability of type checking.  The inner theory can be an easily justifiable extensions of HoTT, allowing the construction of ``infinite structures'' which are considered impossible in plain HoTT. Alternatively, we can set the inner theory to be exactly the current standard formulation of HoTT, in which case our system can be thought of as a type-theoretic framework for working with ``schematic'' definitions in HoTT.  As demonstrations, we define semi-simplicial types and formalise constructions of Reedy fibrant diagrams.
\end{abstract}
\maketitle

\section{Introduction: Motivations for a 2-Level System}

The identity type is probably the single concept of intensional Martin-L\"of type theory ($\MLTT$) which has created most questions, stimulated most research, caused most confusion, and enabled the largest number of different views.
Written $\idtype[A] x y$ for elements $x,y$ of a type $A$, the identity type expresses that two elements are equal in some sense and can be substituted for each other, and the elements of this type are called equalities.
However, by default, it has a somewhat strange standing.
On the one hand, it is not well-behaved when it comes to describing equality of functions and equality of types.
Given two functions of the same type, we cannot derive the principle of (naive) function extensionality, saying that the functions are equal if they are equal at every point, from the basic axioms.
We also cannot show that equivalent types can be substituted for each other, although they do behave equivalently in any given situation.
On the other hand, we also cannot derive the principle of unique identity proofs ($\UIP$): by a construction of Hofmann and Streicher~\cite{hofmannStreicher_groupoids}, we can not show $\idtype[{\idtype x y}] p q$ for two equalities $p,q$. A priori, it is unclear what it should mean to have distinguishable equalities and how one can make sense of this behaviour.

As we view it, there are two major ways to remedy the situation.
Both can be seen as extensions of $\MLTT$. The principle of function extensionality can be added in any case, but after that, we have two possibilities to extend the theory further.
The first is to add $\UIP$ (or, equivalently, Streicher's $K$) as an axiom. Let us call the resulting type theory $\MLTTK$.
The second possibility is to consider univalence, ensuring that type equality is what one would ideally expect.
This approach is taken by Homotopy Type Theory~\cite{HoTTbook}, and we write $\HOTT$ for the resulting theory.

One appeal of $\HOTT$ is that equalities can be seen as paths in a space, and it is even possibly to develop a lot of homotopy theory synthetically.
An important insight is that, when doing homotopy theory in type theory, every statement that we make is up to homotopy, and every construction respects (homotopical) equivalence. 
This means that whatever we do will be ``non-evil'' in the sense that it can only take the homotopy type of spaces, and homotopy equivalence classes of maps, into account, and not the concrete representations of spaces or maps.
Exactly this is often considered a selling point of $\HOTT$: one often defines constructions using representatives of homotopy classes in traditional homotopy theory and is forced to show that the constructions are well-defined, i.e.\ do not depend on the choice of the representative. 
In $\HOTT$, everything we do is automatically well-defined as we are simply not able to talk about strict properties internally.

Going back to the homotopical point of view, it is not hard to imagine that the blessing of having only constructions up to homotopy can turn out to be a curse:
we are unable to make any strict statement.
For example, we cannot form a type expressing that a given diagram commutes strictly; all we can do it stating that it commutes up to homotopy.
Unfortunately, depending on the shape of the diagram, this will only be sufficient in the simplest cases.
More often than not, it will be necessary to say that the different ``pieces'' (the equalities expressing commutativity) fit together. 
For example, the fact that a certain sub-diagram commutes can be part of the proof that the diagram commutes, but it may at the same time be derivable as the composition of the fact that other sub-diagrams commute.
In this case, it is natural to require these different ways of getting a certain proof to be equal.
It does not stop here; these new proofs can itself be required to be coherent, and so on.
What happens here is not at all something that can only be observed in type theory.
The first step becomes already apparent in the theory of monoidal categories in the form of ``Mac Lane's Pengaton''.
On higher dimensions, it is exactly the same issue that is discussed as \emph{homotopy commutativity versus homotopy coherence} by Lurie~\cite{Lurie_higherToposTheory}.

In general, homotopy coherence corresponds to infinite towers of coherence data, and it is a major open problem (and commonly believed to be unsolvable) to express such towers internally in $\HOTT$.
One way to avoid this problem is to restrict constructions to types of low truncation levels.
As an examples, the category theory developed in~\cite{ahrens_rezk} only considers $1$-truncated types and [what corresponds to] ordinary categories.
This is in many situations not satisfactory: we know that types are $\infty$-groupoids~\cite{lumsdaine_phd,bg:type-wkom}, and similarly, the universe should be an $(\infty,1)$-category. 
Unfortunately, we seem to have no way of expressing this internally in $\HOTT$.

The crucial shortcoming of $\HOTT$ is that we are unable to perform some constructions 
which actually \emph{seem} to be harmless as they only require finite amounts of coherences at every step.
An example that has received considerable attention in the $\HOTT$-community is the construction of Reedy fibrant $n$-truncated semi-simplicial types (simply referred to as \emph{semi-simplicial types}).
Let us start with $\deltplus$, the category of finite non-zero ordinals and strictly monotonous functions.
Let us write $\ordinal n$ for the ordinal with $(n+1)$ elements.
A type-valued diagram over $\deltop$ is a strict functor from $\deltop$ to the category of types.
It would correspond to a type $X_{\ordinal n}$ (for simplicity written $X_n$) for every $n$, and face maps $d_i : X_{n+1} \to X_{n}$ for $0 \leq i \leq n$, as it is well-known that any map can be written as a composition of face maps.
The problem is that we need the semi-simplicial identities (essentially a representation of the functor laws) to be strict, which we cannot express in type theory.
The considered approach to avoid this problem is to only attempt internalising Reedy fibrant diagrams over $\deltop$, essentially ensuring that the face maps are simple projections.
Using the correspondence between fibrations and type families, a (Reedy fibrant) semi-simplicial type then 
corresponds to a type $X_0$ (the ``points'') on level $0$. 
On level $1$, we need a family $X_1 : X_0 \to X_0 \to \UU$, where $\UU$ is the universe of types. 
We think of $X_1$ as lines between types.
Next, we need $X_2 : \prd{a,b,c: X_0} X_1(a,b) \to X_1(b,c) \to X_1(a,c) \to \UU$, the type of fillers for triangles. 
Writing down the type of $X_4$ is already a bit tedious, but nevertheless straightforward: $X_4$ is a family which gives a type for any collection of four points, six lines and four triangles that form an empty tetrahedron.
The long-standing open problem of homotopy type theory is to write down the type of $X_n$ in general (up to equivalence).
Perhaps surprisingly, this does not seem to be possible.
What \emph{is} possible is to generate an expression $X_n$ for every externally fixed numeral $n$, such that the expressions $X_0, X_1, X_2, \ldots$ all ``fit together''.
When one tries to do the construction for a \emph{variable} $n : \N$, it does no longer type-check.
The reason is that some judgmental equalities that hold in the case of a numeral $n$ fail to hold in the case of a variable.
We can try to prove in type theory that the required equalities hold up to homotopy.
However, we quickly have to realize that we then also need that these equalities are coherent, and that the coherence proofs are coherent themselves, and so on; something that no one has managed to do so far.
The problem is that we cannot formulate the tower of coherences that we need to prove. 
Morally, the required equalities \emph{should} hold and be fully coherent just because they are trivially satisfied for each externally fixed natural number.
If we can use a system where we judgmental equalities can be shown by induction, there would thus be no problem at all; however, this would require judgmental equalities to be some sort of type.

In $\MLTTK$, the internal equality type can be seen as an internalised version of judgmental equality.
For example, a well-known meta-theoretic statement is that any equality that is constructed in the empty context is $\refl{}$; that is, if we can show an equality internally without assumptions, then this equality holds judgmentally.
Not surprisingly, it is possible to construct Reedy fibrant semi-simplicial types in $\MLTTK$.
However, we can also simply define categories and functors in the naive sense, as all coherences are satisfied automatically.

The idea of a 2-level system is to combine $\MLTTK$ and $\HOTT$ instead of viewing them as two alternative extensions of $\MLTT$.
We can describe this in two ways:
\begin{enumerate}
 \item \label{item:firstKthenEmbed}
 Start with a type theory that has axiom $K$ and consider a ``sub-theory'' of types and maps that do not talk about equalities. 
 Inside this sub-system, we can consider a new equality type and univalent universes. 
 If we use the equality type with $K$ of the outer system, we cannot form types that live in the inner system; however, we can reason about the inner system.
 \item We may start with $\HOTT$ and try to formulate the meta-theory (in which judgmental equality lives) as a type system.
 It is not necessary to capture every aspect of the meta-theory in this type system; the important part is that this outer type system has an equality type (which we call \emph{strict equality}) satisfying $K$.
 We then have in total three equalities: the equality in $\HOTT$; the strict equality; and the judgmental equality (which we should now refer to as \emph{definitional equality}).
 From the point of view of $\HOTT$, the strict equality and definitional equality are identical.
\end{enumerate}

Considering a type theory with two equality types is not new.
Our proposal is motivated by the suggestion of 
a \emph{homotopy type system} (HTS) by Voevodsky~\cite{voe_hts}.
However, as far as we are aware, ``HTS'' mostly refers to a range of ideas so far 
but not to a precise theory, and there is no publication that presents or analyses HTS.
The core idea of HTS is to make some judgmental equalities provable.
In other words, some form of equality reflection, the characteristic concept of extensional type theory, is reintroduced, in a way that is compatible with the ``standard'' intensional identity type.
A concrete theory that could be called ``HTS'' is outlined in the draft~\cite{voe_hts} which, unfortunately, presents rather involved rules that would presumably be non-trivial to justify.
It is sometimes said that \emph{Andromeda}, a project by Bauer et al.~\cite{andromeda}, is an implementation of HTS.
However, this is probably rather misleading as Andromeda goes in a very different direction and does not feature univalence at all.

A key observation of the current paper is that no form of equality reflection is actually required.
Our proposal instead only required unique identity proofs for the strict equality type.
Thus, we can avoid all the problems that are usually connected to equality reflection, such as undecidability of type checking.
In contrast, the theory that we suggest is well-behaved, very close to the standard formulation of $\HOTT$, and has straightforward semantics.
One could expect that a downside of our system might be reduced expressibility compared to a theory that features equality reflection. 
However, we show that we can achieve in our system what HTS was suggested for: a definition of semi-simplicial types, and other constructions.
This should actually not be surprising in the light of Hofmann's result~\cite{hofmann_conservativity}, which states that equality reflection is conservative over $\MLTTK$.

Our 2-level theory can be defined as two separate type theories with a morphism between them. 
This actually gives a recipe for constructing a variety of reasonable 2-level theories, and the choices that can be made affect the exact abilities of the system.
We believe that our 2-level theories can be used in two ways.
First, we can use the outer theory as a powerful formal language to study the inner theory.
For some formulations of the 2-level theory, we get a conservativity result (by an argument of the second-named author; see the forthcoming thesis~\cite{paolo:thesis}).
This means that the inner theory is exactly $\HOTT$ as studied in the standard textbook on homotopy type theory~\cite{HoTTbook} and by many authors.
In a proof assistant which supports this theory, we can then implement results that so far can only be stated meta-theoretically.
To give an example, it is shown in~\cite{kraus_generaluniversalproperty} that constant functions from $A$ to $B$ which satisfy $n$ coherence conditions correspond to maps $\trunc {-1} A \to B$, provided that $B$ is $n$-truncated. 
This can be done in $\HOTT$ only if $n$ is an externally fixed natural number.
In the 2-level system, we can formalise it by taking $n$ to be a number in the outer theory, and show that the equivalence holds in the inner theory.

Second, we can use the construction of 2-level theories to derive extensions of $\HOTT$ that allow constructions that $\HOTT$ does not allow.
For example, we can assume that the natural numbers of the outer theory are \emph{exactly} the natural numbers of the inner theory, something that is satisfied in the simplicial set model.
This gives us a univalent type theory in which various concepts including semi-simplicial types can be defined.

\subparagraph*{Contributions of the paper}
Summarised, the main contributions of the paper are:
\begin{itemize}
 \item We give (for the first time) a clean presentation of a system with two equalities.
 \item Our theory is simple enough to have straightforward semantics. Such semantics have not yet been described for previous proposals~\cite{voe_hts,DBLP:journals/corr/PartL15}.
 \item We demonstrate how our theory allows constructions that are thought to be impossible in standard $\HOTT$, such as semi-simplicial types~\cite{herbelin_semisimpl,DBLP:journals/corr/PartL15}. A partial \textsf{Agda} formalisation is available.
 \item Schematic constructions, which could so far only be given on paper, can be formalised in our system. As an example, we perform constructions of Reedy fibrant diagrams. Further work is outlined in the conclusions.
\end{itemize}

\subparagraph*{Related work}
The current paper is the write-up of our presentation~\cite{altenkirch_twolevels} at TYPES'15.
As briefly explained above, the main difference to Voevosky's draft~\cite{voe_hts} is that we do not consider any form of equality reflection, saving us from various difficulties.

Superficially related is the construction by Maietti~\cite{Maietti2009319} of a \emph{two-level foundation for constructive mathematics}. However, their motivation and goals are very different from ours, hence their system cannot be used to reconcile strict equality with univalence.  

Our work is more closely related to a recent proposal by Part and Luo~\cite{DBLP:journals/corr/PartL15} of Logic-enriched $\HOTT$. 
In their system, our strict layer of type theory is replaced by a ``logic enrichment''.
Their proposal is limited to the construction of semi-simplicial types (corresponding to the one that we give in Section~\ref{sec:semisimp}). 
It is not explained whether this can be generalised to Reedy fibrant diagrams in the sense we present in Section~\ref{sec:diag-inverse}, as they use specific properties of the $\Delta_+$ category.

Herbelin has given a construction of semi-simplicial types along the lines of the one in Section~\ref{sec:semisimp} in an unspecified type theory containing a ``connective'' for strict equality~\cite{herbelin_semisimpl}.

\subparagraph*{\textsf{Agda} formalisation}
As a proof assistant based on our 2-level theory does not (yet) exist, we cannot formalise our constructions exactly as they are presented.
However, we have implemented in Agda an approximation of the construction of semi-simplicial types that is given in Section~\ref{sec:semisimp}.
It can be found on GitHub at \href{https://github.com/nicolaikraus/HoTT-Agda/blob/master/nicolai/SemiSimp/SStypes.agda}{\nolinkurl{github.com/nicolaikraus/HoTT-Agda/tree/master/nicolai/SemiSimp}}.
For an explanation of the relationship between the construction given in the paper and this implementation, we refer to the last remark of Section~\ref{sec:semisimp}.

\subparagraph*{Organisation}
The structure of the paper is as follows.
In Section~\ref{sec:specification}, we specify our 2-level theories.
Section~\ref{sec:semisimp} explains the construction of Reedy fibrant $n$-truncated semi-simplicial types in a way that could \emph{nearly} be done in homotopy type theory, and we show how the missing gap is filled by our strict equality.
Then, in Section~\ref{sec:diag-inverse}, we demonstrate how our theory can be used to internalise standard constructions in a fairly straightforward way. 
Finally, in 
Section~\ref{sec:conclusions}, we outline further work and conclude the paper.

\section{The Specification of a 2-Level System} \label{sec:specification}

In this section, we want to specify our 2-level theory (or, to be precise, our family of 2-level theories). 
We give two presentations: first, the semantical approach, and second, the syntactical approach.
With the first approach, we explain how the theory is constructed.
It also shows which choices can be made, and how models of the 2-level theory can be constructed.
The syntax that we propose afterwards is based on the semantics, but fixes a precise system.

\subsection{Semantical Approach}

Many models of type theory consist of a category $\C$, modelling the category of contexts.
Starting from $\C$, additional structures are added to model types and terms, together with the structure that is needed to model the components of the considered theory (such as universes or dependent functions).
Then, a model of a 2-level theory in our sense is given by a category of contexts $\C$, together with \emph{two} structures on $\C$ such that the first structure (taken together with $\C$) models $\HOTT$, and the second structure (taken together with $\C$) models $\MLTTK$.
Finally, we need a morphism between the structures in a suitable sense, describing how any type or term in $\HOTT$ can be viewed as a type or term in $\MLTTK$. 

We make this precise using the notion of \emph{categories with families}~\cite{dybjer1995internal}.  Let us recall the definition:

\begin{definition}[CwF {\cite{dybjer1995internal}}]
A \emph{category with families} (CwF) is given by:
\begin{itemize}
\item a category $\C$, equipped with a distinguished terminal object $1$;
\item a presheaf $\Ty : \C \to \op{\setcat}$;
\item a presheaf $\Tm : \int \Ty \to \op{\setcat}$;
\item for all $\Gamma : \C$ and $A : \Ty(\Gamma)$, an object $(\Gamma.A, \pi_A)
  : \C/\Gamma$ representing the functor $\C/\Gamma \to \op{\setcat}$ defined by:
$$
(\Delta, \sigma) \mapsto \Tm_\Delta(A[\sigma]).
$$
\end{itemize}

The objects of $\C$ are called \emph{contexts}. Given a context $\Gamma$, the
elements of $\Ty(\Gamma)$ are called \emph{types}, and given a type $A$, the
elements of $\Tm_\Gamma(A)$ are called \emph{terms}.

The context $\Gamma.A$ is called the \emph{context extension} of $\Gamma$ by the
type $A$, and $\pi_A$ is the \emph{display map} of $A$.

The action of $\Ty$ and $\Tm$ on morphisms is called \emph{substitution}.
\end{definition}

A CwF can be regarded as a model of $\MLTT$ with only
\emph{structural rules}, i.e.\ rules that deal with types, terms and
substitutions, but no type formers (like $\Pi$ or $\Sigma$ types).  Type formers
can be postulated separately as additional structures on a CwF.
For details, we refer to \cite{dybjer1995internal} and \cite{hofmann_syntaxSemantics}.

To model a 2-level type theory, we need to add some extra structure to a CwF:
\begin{definition}\label{def:cwf2}
A \emph{2-level category with families} is a CwF $\C$, together with:
\begin{itemize}
\item a presheaf $\fTy : \C \to \op{\setcat}$;
\item a natural transformation $|-|: \fTy \to \Ty$.
\end{itemize}
\end{definition}

Given a 2-level CwF $\C$, we will denote the underlying category with family by
$\strict{\C}$.  There is also a second CwF structure on $\C$, where the types
are given by $\fTy$, and terms are defined as:
$$
\fTm_{\Gamma}(A) = \Tm_{\Gamma}(| A |),
$$
and context extension is given simply by $\Gamma.A = \Gamma.| A |$.  We will
denote this second CwF by $\fibrant{\C}$.

The map $|-|$ determines a morphism of CwF $\fibrant{\C} \to \strict{\C}$,
which we will also denote by $|-|$.

The theory employed in this paper will be modelled by a 2-level CwF $\C$ where:
\begin{itemize}
\item $\strict{C}$ is a model of $\MLTTK$;
\item $\fibrant{C}$ is a model of $\HOTT$;
\item the morphism $|-|$ preserves $\Pi$, $\Sigma$ and $1$ \emph{strictly}.
\end{itemize}

Note that, crucially, equality types, although present in both CwF structures,
are \emph{not} generally preserved.  This is important, because preservation of
equality would mean that axiom $K$ holds in $\fibrant{\C}$, which in turn would
imply that $\fibrant{\C}$ does not admit any univalent universes containing
non-propositional types.

Other type formers besides those mentioned might or might not be preserved.  We
say that a 2-level CwF is \emph{strong} if $|-|$ preserves coproducts, natural numbers, and the empty type (more generally $W$-types, if part of the theory). 

Interestingly, most of the existing models of $\HOTT$ can be naturally extended
to a 2-level CwF.  Most notably, the simplicial model
\cite{kapLumsVoe_ssModelOfUF} can be regarded as a 2-level CwF, where $\Ty$ is
given by arbitrary (well-ordered) morphisms, $\fTy$ is the subfunctor of $\Ty$
consisting of Kan fibrations, and $|-|$ is simply the inclusion.  With this
setup, $\strict{C}$ is (equivalent to) a presheaf CwF, which models type theory
with equality reflection (hence, in particular, $\MLTTK$), and $\fibrant{C}$ is
the same as the model defined in the paper.

One can also start with an arbitrary model $\C$ of $\HOTT$, then consider the
presheaf category $\presheaf{\C}$.  It is perhaps not surprising that one can
equip $\presheaf{\C}$ with a 2-level CwF structure so that $\C$ can be recovered
inside $\fibrant{\presheaf{C}}$.  This makes it possible to use 2-level type
theory to formulate and prove statements that hold in any model of $\HOTT$,
i.e.\ 2-level type theory can be regarded as a meta-language for $\HOTT$.

However, the details of this construction are somewhat involved, mainly due to
the strictness requirement in Definition~\ref{def:cwf2}.  Therefore, we will
not explore that direction further in this paper and refer instead to the
forthcoming thesis of the second-named author~\cite{paolo:thesis}.

\newcommand{\steq}{\stackrel{s}{=}}
\newcommand{\ctx}{\mathrm{ctx}}

\subsection{Syntactical Approach}

In the syntactical approach, the clear separation of a 2-level theory into two theories becomes harder to see.
We do not leave as many choices open as in the semantical approach, but rather fix a concrete theory; and the choices that we make ensure that the conservativity result of the forthcoming thesis~\cite{paolo:thesis} applies to the presented theory.

For a precise specification, we choose a presentation in the style of~\cite[Appendix A.2]{HoTTbook}, which considers three forms of judgments: $\Gamma \vdash \ctx$; $\Gamma \vdash a : A$; and $\Gamma \vdash a \jdeq a' : A$.
Fortunately, we do not need to give \emph{all} the rules, as most of them are identical to those given in~\cite[Appendix A.2]{HoTTbook}. 
Thus, in most cases, it is sufficient to state the difference in order to give both an understandable and a precise specification.

\newcommand{\strictN}{\N^s}

The theory that we consider has the following basic types and type formers: 
$\Pi$, the type former of dependent functions;
$\Sigma$, the type former of dependent pairs;
$+$, the coproduct type former;
$\unit$, the unit type;
$\emptyt$, the empty type;
$\N$, the fibrant type of natural numbers;
$=$, the equality type (in the sense of $\HOTT$);
a hierarchy $\UU_0, \UU_1, \ldots$ of universes.
So far, we can think of these as the types and type formers of $\HOTT$.
Further, we have:
$+^s$, the strict coproduct;
$\emptyt^s$, the strict empty pretype;
$\strictN$, the strict pretype of natural numbers;
$\steq$, the strict equality;
and hierarchy $\UU_0^s, \UU_1^s, \ldots$ of strict universes.

Both the hierarchy $\UU_0, \UU_1, \ldots$ and the hierarchy $\UU_0^s, \UU_1^s, \ldots$ are cumulative.
We think of the elements of $\UU_i$ as \emph{fibrant types} (or simply \emph{types}), while the elements of $\UU_i^s$ are \emph{pretypes}.

Recall possibility \ref{item:firstKthenEmbed} from the two ways of describing a 2-level system as outlined on page \pageref{item:firstKthenEmbed}: we can start with a type theory with $K$ and embed $\HOTT$ later.
Thus, we first consider the type theory with the basic types $\emptyt^s$, $\unit$, $\strictN$, with universes $\UU_0^s, \UU_1^s, \ldots$, and with $+^s$, $\Pi$, and $\Sigma$.
All rules correspond exactly to those of~\cite[Appendix A.2]{HoTTbook}.
For example:
\begin{itemize}
 \item Contexts are formed using elements of $\UU_i^s$, i.e.\ if $\Gamma$ is a context and $\Gamma \vdash A : \UU_i^s$, then $\Gamma.A$ is a context.
 \item If $\Gamma \vdash A : \UU_i^s$ and $\Gamma.A \vdash B : \UU_i^s$, then we have $\Gamma \vdash \prd A B : \UU_i^s$.
 \item If $\Gamma,x:A \vdash b : B$, then we have $\Gamma \vdash \lam x b : \prd A B$.
 \item All further rules of $\Pi$, and all rules of $\Sigma$, $+^s$, $\emptyt^s$, $\unit$, and $\strictN$ are also those given in~\cite[A.2.4--9]{HoTTbook}. 
 The constructors of $+^s$ are called $\inl^s$, $\inr^s$, and the constructors of $\strictN$ are called $\zero^s$ and $\suc^s$.
 We assume all the usual judgmental rules (including the judgmental $\eta$-rule for $\Sigma$).
\end{itemize}
Further, the theory has a strict identity pretype, written $\steq$: For any $\Gamma \vdash A : \UU_i^s$ and $\Gamma \vdash a_1,a_2 : A$, we have $\Gamma \vdash a_1 \steq a_2 : \UU_i^s$, with the introduction rule $\refl{}^s$, the eliminator $J^s$, and the usual computation rule. 
For pretypes $A,B : \UU_i^s$, we can form the pretype of strict isomorphisms, written $A \simeq^s B$ (unlike in $\HOTT$, it is enough to have maps in both directions such that both compositions are pointwise strictly equal to the identity).
However, we do \emph{not} assume that $\UU_i^s$ is univalent. 
Instead, we add the rule $K^s$: for $A$, $a_1$, $a_2$ as before, and for $\Gamma \vdash p,q : a_1 \steq a_2$, we have a term $\Gamma \vdash K^s(p,q) : p \steq q$. 
We also assume that $\steq$ satisfies the principle of function extensionality.

Note that, so far, we have not considered $\UU_i$, $+$, $\emptyt$, $\N$, $=$ at all.
We do this now, and their rules are more subtle.
The first important rule is that any type (element of $\UU_i$) is also a pretype (element of $\UU_i^s$):
\begin{equation}\label{eq:type-is-pretype}
 \infer[]{\Gamma \vdash A : \UU_i^s}{\Gamma \vdash A : \UU_i}
\end{equation}
This means that informally we can understand $\UU_i$ as a subtype of $\UU_i^s$.

Now, let $A$ and $B$ be fibrant types, i.e.\ $\Gamma \vdash A : \UU_i$ and $\Gamma.A \vdash B : \UU_i$.
Then, by~\eqref{eq:type-is-pretype} and by the formation rule of $\Pi$, we have $\Gamma \vdash \prd A B : \UU_i^s$.
However, we add the rule that, under these conditions, this conclusion can be lifted to $\Gamma \vdash \prd A B : \UU_i$. 
In other words, $\Pi$ preserves types.
We add the same rule for $\Sigma$:
\begin{equation}
 \infer[]{\Gamma \vdash \smsimple A B : \UU_i}{\Gamma \vdash A : \UU_i & \Gamma.A \vdash B : \UU_i}
\end{equation}
We do \emph{not} add the same rule for $+^s$, that is, the strict sum of two types is still only a pretype.
Similarly, there is no special rule for $\steq$: if $\Gamma \vdash a_1,a_2 : A$, it does not matter whether $A$ is a type or only a pretype, the expression $a_1 \steq a_2$ is only an element of $\UU_i^s$, not of $\UU_i$.

In contrast, the equality type former $=$ can only be applied to elements of fibrant types; i.e.\ its formation rule is given by
\begin{equation}
 \infer[]{\Gamma \vdash a_1 = a_2 : \UU_i}{\Gamma \vdash A : \UU_i & \Gamma \vdash a_1,a_2 : A}
\end{equation}
(note that there is no strict universe $\UU_i^s$ involved).
The type $a_1 = a_2$ (with the constructor $\refl{}$) is a pretype by rule~\eqref{eq:type-is-pretype}, but (usually) not the same as $a_1 \steq a_2$.
The elimination principle of $=$ only works for families of types (not in general for pretypes).
This means that the usual ``path induction'' principle, which allows us to construct an element of $\prd{a_1,a_2 : A}\prd{p:a_1 = a_2}P(a_1,a_2,p)$, can only be applied if $P$ is a family of types, i.e.\ $\Gamma \vdash P : (\smsimple{a_1,a_2:A}a_1=a_2) \to \UU_i$.
If we restrict ourselves to types, we can do everything that we can do in $\HOTT$. 
In particular, we can say what it means for a function between types to be an equivalences (using $=$).
We assume that the universes $\UU_0, \UU_1, \ldots$ are univalent, that is, the canonical map from type of equalities $A = B$ to the type of equivalences $A \simeq B$ (defined as usual in homotopy type theory) is an equivalence itself.

Similarly, the type former $+$ only allows us to form a type $A+B$ if $A$ and $B$ are types (elements of some $\UU_i$), and we can only defined a function $\prd{x:A+B}P(x)$ with the usual induction principle if $P$ is a family of types.

We have the type of natural numbers $\N : \UU_0$ (in any context) with the constructors $\zero$, $\suc$, and its induction principle can only be applied to eliminate into families of types.
The same is the case for $\emptyt$.
This completes the syntactical characterization of our 2-level system.
We will usually omit the index and simply write $\UU^s$ or $\UU$ instead of $\UU_i^s$ or $\UU_i$ in the same style as it is done in~\cite{HoTTbook}.
A strong 2-level theory is now simply one in which $\emptyt^s$ and $\emptyt$, and $+^s$ and $+$, and $\strictN$ and $\N$ coincide.

\begin{remark}
 If $A$ is a (``fibrant'') type with elements $a_1,a_2 : A$, then we can form both the type $a_1 = a_2$ and the pretype $a_1 \steq a_2$. 
 By ``strict path induction'' (i.e.\ an application of $J^s$), we can easily construct a function $a_1 \steq a_2 \to a_1 = a_2$.
 Consequently, strictly equal elements of a type are also homotopy-equal.
 This corresponds to the fact that judgmental equality in $\HOTT$ implies equality (``$\refl{}$'').
 We cannot construct a function in the other direction, as the path induction principle $J$ can only be applied to eliminate into types, which $a_1 \steq a_2$ is not. 
 Hence, equal elements are not necessarily strictly equal.
 However, if we have a type which does satisfy this ``equality reflection'' principle, it is easy to see that the type is a set in the sense of homotopy type theory.
\end{remark}

\newcommand{\sst}{\mathsf{SST}}
\newcommand{\ssx}{\mathsf{SS}}
\newcommand{\ssk}{\mathsf{SK}}
\newcommand{\sskmor}{\mathsf{SK}^\to}
\newcommand{\ssfunclaw}{\mathsf{\alpha}}

\section{Semi-Simplicial Types} \label{sec:semisimp}

In a 2-level theory, we can define strict categories in a reasonable sense. 
There are a number of choices that one can make; 
for example, the objects could be a fibrant type or only assumed to be a pretype.
Later (see Definition~\ref{def:strictcat}), we will give one possible concrete definition.
The important thing is that the categorical equations can be required strictly; and, if we have such a strict category $\mathcal C$, we can easily write down the pretype of strict functors $\mathcal C \to \UU$.
Unfortunately, there is no general way to get an actual fibrant type of such functors.

The case where $\mathcal C$ is $\Delta^\mathrm{op}$ (the category of finite nonempty ordinals and increasing functions) is particularly interesting since ``simplicial structures'' appear frequently in homotopy theory.
Having a type of functors $\Delta^\mathrm{op} \to \UU$ would have many potential applications; 
maybe most notably, one could try to internalise a constructive version of the model of univalent foundations in simplicial sets~\cite{kapLumsVoe_ssModelOfUF}.
Unfortunately, it seems unreasonable to expect that such a type can be constructed.
It would be a good approximation (and potentially good enough for many constructions) if one could form a type of functors $\deltop \to \UU$, where $\deltplus$ is the category of finite nonempty ordinals and \emph{strictly} increasing functions (a more precise definition will be given below).
Trying to define such a type seems more promising, since $\deltop$ is an \emph{inverse category} and, if we restrict ourselves to \emph{Reedy fibrant} functors, we can describe them by induction (see~\cite{shulman_inversediagrams}).

This gave rise to the challenge of defining Reedy fibrant $n$-truncated semi-simplicial types (in the community often just referred to as \emph{semi-simplicial types}) in type theory.
The challenge was first raised during the special year on Univalent Foundations at the Institute for Advanced Study (Princeton, 2012--13).
As briefly sketched in the introduction, a (Reedy fibrant) $0$-truncated semi-simplicial type is simply a type $X_0 : \UU$, a $1$-truncated semi-simplicial type is such an $X_0$ together with a family $X_1 : X_0 \to X_0 \to \UU$, and so on.
Defining $n$-truncated semi-simplicial types as a family $\sst : \N \to \UU$ in homotopy type theory is a famous open problem.
Many attempts (see e.g.\ \cite{herbelin_semisimpl,shulman:eating,nicolai:thesis}) 
have not led to a solution, and at a workshop on HoTT in Warsaw (June 29--30, 2015), a clear majority of the participants expected it to be impossible.

The hard part of the construction is to define the \emph{matching objects} $M_n$, that is the ``full boundary'' of an $n$-simplex, as the corresponding component of $\sst_n$ is then just given as a family $M_n \to \UU$.
A popular attempt for defining the matching objects $M_n$ is to define the \emph{$k$-skeleton} $\ssk_n^k$ of $\sst_n$ by induction on $k$, that is, the collection of components of $\sst_n$ up to level $k$.
As long as $k$ is a fixed numeral, this can be done.
However, if $k$ is a variable, some crucial judgmental equalities do not hold anymore
and the construction is believed to become impossible.
In our $2$-level theory, we can prove strict equalities (i.e.\ the internalisation of judgmental equality) by induction.
This allows us to complete the sketched approach of defining $\sst$ in a weak sense:
we construct a family $\sst : \strictN \to \UU$. 
If we assume that the strict natural numbers ($\strictN$) and the fibrant ones ($\N$) coincide, this represents a construction of $n$-truncated semi-simplicial types.
Without this assumption and under the conjectured conservativity result~\cite{paolo:thesis},
it internalises the result that $n$-truncated semi-simplicial types can be defined for an externally fixed $n$.

To give the precise construction, let us first note that we have the family $\Fin : \strictN \to \UU$ of finite types ($\Fin_n$ is the type with $n$ elements), together with the families $<^n : \Fin_n \to \Fin_n \to \UU$.
Let us write $\isIncr_{i,j}$ for the predicate 
 \begin{align*}
  &\isIncr_{i,j} : (\Fin_i \to \Fin_j) \to \UU \\
  &\isIncr_{i,j} (f) \defeq \prd{x,y : \Fin_i} (x <^i y) \to (f(x) <^j f(y)),
 \end{align*}
expressing that a function is strictly monotonously increasing.
Let us further write $\deltplus(i,j)$ for the type $\sm{f : \Fin_{i} \to \Fin_{j}} \isIncr_{i,j}(f)$.
We then have a composition operator $\circ : \deltplus(h,i) \to \deltplus(i,j) \to \deltplus(h,j)$, defined separately on each of the two components.
This is a representation of strictly increasing functions such that $\circ$ is strictly associative, as observed in~\cite{nicolai:thesis}.
Unsurprisingly, this is enough to make $\deltplus$ a \emph{category} in the sense that we will define later (see Definition~\ref{def:strictcat}).\footnote{Note that, for technical reasons, we include the initial object $\Fin_0$. This explains the shift by $1$: we have defined $\deltplus(i,j) \defeq \sm{f : \Fin_{i} \to \Fin_{j}} \isIncr_{i,j}(f)$ instead of $\sm{f : \Fin_{i+1} \to \Fin_{j+1}} \isIncr_{i,j}(f)$.}

In the following, we use variable names $\vec X$, $\vec x$ instead of $X$, $x$ to indicate that we have an element of a nested $\Sigma$-type, i.e.\ a tuple.
With $\deltplus$ at hand, we define truncated semi-simplicial types ($\sst$) simultaneously with skeletons ($\ssk$) and the morphism part of skeletons (written $\sskmor$).
These have the following types:
\begin{equation*}
 \begin{alignedat}{3}
  & \sst && : \strictN \to \UU && \hspace*{-2cm} \text{--- we write $\sst_k$ instead of $\sst(k)$;} \\
  & \ssk && : \prd{k : \strictN} \; \sst_k \to \strictN \to \UU && \hspace*{-2cm} \text{--- we write $\ssk_{k,\vec X}^{n}$ instead of $\ssk(k,\vec X,n)$;} \\
  & \sskmor && : \prd{k : \strictN} \prd{\vec X : \sst_k} \prd{m,n : \strictN} \prd{f : \deltplus(m,n)} \ssk_{k,\vec X}^n \to \ssk_{k,\vec X}^m \\
  &&&&& \hspace*{-2cm} \text{--- we write $\sskmor_{k,\vec X}$ instead of $\sskmor(k,\vec X,m,n)$.}
 \end{alignedat}
\end{equation*}
These type families can be explained as follows:
\begin{enumerate}
 \item $\sst_k$ is the type of $(k-1)$-truncated semi-simplicial types.
 \item Assume we have a $(k-1)$-truncated semi-simplicial type $\vec X$, where $k$ is smaller than another given number $n$.\footnote{The definition works for $k \geq n$, but $k < n$ is the case that is important for the intuition.} 
$\vec X$ allows us to form the type $\ssk_{k,\vec X}^n$.
This is the type of ``partial boundaries'' of an $(n-1)$-truncated semi-simplicial type.
Intuitively, it has $n$ points, $\binom n 2$ lines, \ldots, and $\binom n k$ cells on level $(k-1)$.
 \item We think of $\ssk_{k,\vec X}$ as a ``functor'' from $\deltplus$ to $\UU$. 
Its morphism component is given by $\sskmor_{k,\vec X}$: for any $f : \deltplus(m,n)$, we get a function $\ssk_{k,\vec X}^n \to \ssk_{k,\vec X}^m$ which simply ``removes'' those cells that appear in the partial boundary of an $(n-1)$-simplex, but not in the partial boundary of an $(m-1)$-simplex.
\end{enumerate}
At the same time as we define $\sst$, $\ssk$ and $\sskmor$, we prove the following strict functor law for all $k, l, m, n : \strictN$, $\vec X : \sst_k$, and $f : \deltplus(l,m)$, $g : \deltplus(m,n)$:
\begin{equation*}
 \ssfunclaw_k (\vec X,f,g) : \sskmor_{k, \vec X} g \circ \sskmor_{k,\vec X} f \; \steq \; \sskmor_{k,\vec X}(g \circ f).
\end{equation*}

We define all the components by induction on $k$ as follows.
In the base case, we set
$\sst_0 \defeq \unit$; $\ssk_{0,\ttt}^n \defeq \unit$; $ \sskmor_{0,\ttt}f \defeq  \idfunc [\unit]$; and $\ssfunclaw_0(\ttt,f,g) \defeq \refl{}^s$.
In the successor case, we choose
\begin{equation*}
 \begin{alignedat}{2}
  & \sst_{k+1} && \defeq \sm{\vec X : \sst_k} \ssk_{k,\vec X}^{k+1} \to \UU \\
  & \ssk_{k+1,(\vec X,Y)}^n && \defeq \sm{\vec x : \ssk_{k,\vec X}^n}  \prd{f : \deltplus(k+1,n)}  Y\left(\sskmor_{k,\vec X}(f, \vec x)\right) \\
  & \sskmor_{k+1,(\vec X,Y)}(f , (\vec x , h)) && \defeq  \left( \sskmor_{k,\vec X} (f,\vec x) \; , \; \lam g \trans {\ssfunclaw_k} {h(f\circ g)} \right) \label{eq:sskmor}\\
 \end{alignedat}
\end{equation*}
Note that, in the last line, the type of the term $h(f \circ g)$ is $Y(\sskmor_{k, \vec X}(g \circ f , \vec x))$.
However, what we need at that point is an element of the type $Y(\sskmor_{k, \vec X}(g, \sskmor_{k, \vec X}(f,\vec x)))$.
This is why we transport along the proof $\ssfunclaw_k(\vec X , f,g)$, abbreviated to $\ssfunclaw_k$, which shows that the two types are strictly equal.

We omit the term for $\ssfunclaw_{k+1}$ as it is not insightful to write it down explicitly.
It is constructed as follows.
First, we note that we need a (strict) equality between pairs; the first components are (strictly) equal by $\ssfunclaw_k$.
When one tries to prove that the second components are (strictly) equal, one quickly realizes that what is needed is coherence for the family of strict equalities $\ssfunclaw_k$:
The composition $\sskmor_{k,\vec X} g \circ \sskmor_{k,\vec X} f \circ \sskmor_{k,\vec X} e$ can be shown to be (strictly) equal to $\sskmor_{k,\vec X} (g \circ f \circ e)$ in two ways, and we need that both ways are strictly equal.
Of course, this follows from the fact that we have axiom $K$ for our strict equality.
We have verified this construction in Agda (see the remark below). 

\begin{remark}
 We and many others have attempted to formalize semi-simplicial types in homotopy type theory with exactly the outlined strategy, replacing the strict law $\ssfunclaw$ by the usual equality type.
 This works in the same way until the point where we need that $\ssfunclaw_k$ is coherent, which is automatic in our case.
 Intuitively, $\ssfunclaw_k$ \emph{is} coherent, and it is easy to get trapped into thinking that this coherence can just be shown simultaneously with the other four components.
 However, if one does this, one notices that one needs an additional coherence level for $\ssfunclaw_{k-1}$,
 and it continues like this.
 Morally, all these coherences should hold, and it is very likely that we would actually be able to prove them inductively if only we were able to write them down.
 Unfortunately, writing them down is a problem that is very similar to defining semi-simplicial types itself.
 From this point of view, what the $2$-level theory gives us is the possibility to prove a certain equality and \emph{all} its higher coherences at the same time. 
\end{remark}
\begin{remark}
 We have now defined the $n$-truncated semi-simplicial types $\sst : \N_s \to \UU$, so we may ask whether we can define a (non-truncated) type of semi-simplicial types $\ssx : \UU$.
 If we work in the strong 2-level theory (where $\N_s$ and $\N$ coincide), 
 we can consider the homotopy limit $\ssx : \UU$, defined as
  \begin{align*}
    \ssx & \defeq \sm{f : \prd{n:\N} \sst_n} \prd{i : \N} \fst(f(i + 1)) = f(i)
  \end{align*}
 Then, $\ssx$ is indeed a (fibrant) type that encodes (Reedy fibrant) functors $\deltop \to \UU$.
\end{remark}
\begin{remark}
 Our Agda formalisation\footnote{\href{https://github.com/nicolaikraus/HoTT-Agda/blob/master/nicolai/SemiSimp/SStypes.agda}{\nolinkurl{github.com/nicolaikraus/HoTT-Agda/tree/master/nicolai/SemiSimp}}} takes place within the fibrant theory. 
 The contribution of the strict equality is completely encapsulated in a single lemma that we postulate without a formal proof.
 Unfortunately, simulating our 2-level system completely in Agda, although possible in principle, would be extremely cumbersome because of the need to keep track of type fibrancy manually. 
\end{remark}

\section{Reedy Fibrant Diagrams Over Inverse Categories} \label{sec:diag-inverse}

In Section~\ref{sec:semisimp}, we have defined Reedy fibrant truncated semi-simplicial types using our 2-level theory.
We have stayed in the fibrant theory ($\HOTT$) as much as we could, and only used the strict theory to prove a crucially needed coherence. 
In this section we want to demonstrate that the 2-level theory is even more powerful if we give up this strategy of only working in the fibrant fragment whenever possible.
The point is that we can derive results about $\HOTT$ without staying inside $\HOTT$, 
analogous to how one can get results that respect homotopy equivalence even when certain constructions are performed on concrete spaces that only represent homotopy types.

What we claim is that, in a proof assistant implementing a 2-level type theory, we could formalize many constructions that are presented meta-theoretically in the current literature.
In the current section, we will show that Reedy fibrant diagrams $I \to \UU$ have limits in $\UU$ if $I$ is a finite inverse category. This is an internalised version of results discussed by Shulman~\cite{shulman_inversediagrams}.
Of course it generalizes the construction in Section~\ref{sec:semisimp}, although not ``literally'': the truncated semi-simplicial types that we get here will look different from those constructed above.

\subsection{Essentially Fibrant Pretypes and Strict Fibrations}

As a preparation for our ``more abstract'' sample applications of the 2-level theory, we remark that it is often not necessary to know that a pretype $A : \UU^s$ is a fibrant type. 
Instead, it is usually sufficient to have a fibrant type $B : \UU$ and a strict isomorphism $A \simeq^s B$. 
If this is the case, we say that $A$ is \emph{essentially fibrant}. 
Clearly, every fibrant type is also an essentially fibrant pretype.

In Section~\ref{sec:semisimp}, we have made heavy use of the fibrant finite types $\Fin_n$ (for $n : \N$).
In a strong 2-level theory, this type coincides with the strict pretype $\Fin_n^s : \UU^s$ (for $n : \strictN$), but this is not in general the case.
We say that some pretype $I$ is \emph{essentially finite} if we have a number $n : \strictN$ and a strict isomorphism $I \simeq^s \Fin_n^s$.
\begin{lemma}
 Let $I$ be essentially finite and $X : I \to \UU$ be a family of fibrant types. Then, $\prd{i:I} X(i)$ is essentially fibrant.
\end{lemma}
\begin{proof}
 Essential finiteness gives us a cardinality $n$ on which we do induction. If $n$ is $\zero^s$, then $\prd{i:I} X(i)$ is strictly isomorphic to the unit type. 
 Otherwise, we have an essentially finite $I'$ such that $f : \unit +^s I' \simeq^s I$, and $\prd{i : I} X(i)$ is strictly isomorphic to
$X(f(\mathsf{inl}\ \ttt)) \times \prd{i : I'} X(f(\mathsf{inr}\ i))$,
 which is essentially finite by the induction hypothesis.
\end{proof}

Similar to essential fibrancy, we have the following definition:

\begin{definition}[strict fibration]
Let $p : E \to B$ be a function (with $E, B : \UU^s$).  
We say that $p$ is a \emph{strict fibration} if we have a family $F : B \to \UU$ such that
the fibre of $p$ over any $b : B$ is strictly isomorphic to $F(b)$, that is,
$\prd{b:B} \left(F(b) \simeq^s \sm{e:E}p(e) \steq b \right)$. 
\end{definition}

From now on, we will drop the attribute \emph{strict} and simply talk about \emph{fibrations}.
Any fibrant type family $F : B \to \UU$ gives rise to a fibration $p : E \to B$, 
as it is easy to see that the first projection $(\smsimple B F) \to B$ satisfies the given condition.
Indeed, any strict fibration is isomorphic over $B$ to a strict fibration of this form.
This often allows us to assume that a given fibration has the form of a projection.

\subsection{Strict Categories}

We define categories in much the same way as the precategories are defined in~\cite{HoTTbook}, except that we
use strict equality to express the laws. 
Since strict equality does not suffer from coherence issues, this notion of category is well-behaved. 
It can be applied to structures which do not have fibrant types of objects or morphisms.

\renewcommand{\C}{\mathcal C}
\let\seq\steq 
\newcommand{\obj}[1]{\vert #1 \vert}

\begin{definition}[strict category] \label{def:strictcat}
A \emph{strict category} $\C$ is given by: a pretype $\obj \C : \UU^s$ of \emph{objects};
for all pairs $x, y : \obj \C$, a pretype $C(x, y) : \UU^s$ of \emph{arrows}
or \emph{morphisms};
an \emph{identity} arrow 
$\mathsf{id} : \C(x, x)$ for every object $x$; and a \emph{composition} function
$\circ : \C(y, z) \to \C(x, y) \to \C(x, z)$ for all objects $x,y,z$.
The usual categorical laws are required to hold strictly, that is, 
we have strict equalities $f \circ \mathsf{id} \steq f$ and $\mathsf{id} \circ f \steq f$, as well as $h \circ (g \circ f) \steq (h \circ g) \circ f)$.

We say that a category is \emph{essentially finite} if the pretype of objects $\obj \C$ is essentially finite (no condition is put on the arrows).
\end{definition}

The usual theory of categories can be reproduced in the context of strict
categories.  We leave it to the reader to define appropriate notions of \emph{functor},
\emph{natural transformation}, \emph{limits}, \emph{adjunctions}, and so on.

From now on, we will refer to strict categories simply as \emph{categories}. If
$\C$ is a category, we will often abuse notation and use $\C$ itself to denote
its type of objects.

Another important notion is the following:
\begin{definition}[reduced coslice]
Given a category $\C$ and an object $x : \C$, the \emph{reduced coslice}
$x \sslash \C$ is the full subcategory of non-identity arrows in the coslice
category $x \slash \C$. 
A concrete definition is the following.
The objects of $x \sslash \C$ are triples 
of an $y : \obj \C$, a morphism $f : \C(x,y)$, and a proof $\neg \left(\trans p
f \steq \mathsf{id} \right)$, for all $p : x \steq y$,
where $\transf p$ denotes the $\mathsf{transport}$ function $\C(x,y) \to \C(y,y)$.
Morphisms between $(y,f,s)$ and $(y',f',s')$ are elements $h : \C(y,y')$ such that $h \circ f \steq f'$ in $\C$.

Note that we have a ``forgetful functor'' $\mathsf{forget} : x \sslash \C \to \C$, given by the first projection on objects as well as on morphisms.
\end{definition}

\subsection{Inverse Categories}

Classically, \emph{inverse categories} are categories which do not contain an infinite sequence of nonidentity arrows (see~\cite{shulman_inversediagrams}).
We restrict ourselves to those which have \emph{height} at most $\omega$, and where a \emph{rank function} is given explicitly.
First, consider the category $\oppo{{\strictN}}$ which has $n : \strictN$ as objects, and $\oppo{{\strictN}}(n,m) \defeq n >^s m$ (the function $>^s : \strictN \to \strictN \to \UU^s$ is defined in the canonical way).
Then, we define:
\begin{definition}[inverse category]
 We say that a category $\C$ is an \emph{inverse category}
 if there is a functor $\varphi : \C \to \oppo{{\strictN}}$ which reflects identities; i.e.\ if we have $f : \C(x,y)$ and $\varphi_x \steq \varphi_y$, then we also have $p : x \steq y$ and $\trans p f \steq \mathsf{id}$.
\end{definition}

\subsection{Reedy Fibrant Limits}

Much of what is known about the category of sets in classical category theory
can be extended to the category of pretypes in a given universe. 
For example, the following result translates rather directly:

\begin{lemma} \label{lem:all-strict-limits}
The universe $\UU^s$, viewed as a category in the canonical sense, has all small limits.
\end{lemma}
\begin{proof}
Let $\C$ be a category with $\obj \C : \UU^s$ and $\C(x,y) : \UU^s$ (for all $x,y$).
Let $X : \C \to \UU^s$ be a
functor. 
We define $L$ to be the pretype of natural transformations $\mathsf 1 \to X$, where 
$\mathsf 1 : \C \to \Type$ is the constant functor on $\unit$.  
Clearly, $L : \UU^s$, and
a routine verification shows that $L$ satisfies the universal property of the
limit of $X$.
\end{proof}

Unfortunately, the category $\UU$ of fibrant types is not as well behaved.  
Even pullbacks of fibrant types are not fibrant in general (but see Lemma~\ref{prop:fibrant-pullback}). 
If we have a functor $X : \C \to \UU$, we can always regard it as a functor $X : \C \to \UU^s$, where it does have a limit.
If this limit happens to be essentially fibrant, we say that $X$ has a \emph{fibrant limit}.
Clearly, this limit will then be a limit of the original diagram $C \to \UU$ (note that $\UU$ is a full subcategory of $\UU^s$).

\begin{lemma}\label{prop:fibrant-pullback}
The pullback of a fibration $E \to B$ along any function $f : A \to B$ is a fibration.
\end{lemma}
\begin{proof}
We can assume that $E$ is of the form $\sm{b:B} C(b)$ and $p$ is the first projection.
Clearly, the first projection of $\sm{a:A}C(f(a))$
satisfies the universal property of the pullback.
\end{proof}

Lemma \ref{prop:fibrant-pullback} makes it possible to construct fibrant limits of
certain ``well-behaved'' functors from inverse categories.
The so-called \emph{matching objects} play an important role.
\begin{definition}[matching object; see {\cite[Chp.~11]{shulman_inversediagrams}}]
Let $\C$ be an inverse category, and $X : {\C} \to \UU$ a functor. 
For any $z : \C$, we define the \emph{matching object} $M_z^X$ to be the (not necessarily fibrant) limit of
the composition
$z \sslash \C \xrightarrow{\mathsf{forget}} \C \xrightarrow{X} \UU \subset \UU^s$.
\end{definition}

\begin{definition}[Reedy fibrant diagram; see {\cite[Def. 11.3]{shulman_inversediagrams}}]
 Let $\C$ be an inverse category and $X : \C \to \UU$ be a functor.
 We say that $X$ is \emph{Reedy fibrant} if, for all $z : \C$, the canonical map $X_z \to M_z^X$ is a fibration.
\end{definition}

Using this definition, we can make precise the claim that we can construct fibrant limits of certain well-behaved diagrams:
\begin{theorem}[see {\cite[Lemma 11.8]{shulman_inversediagrams}}] \label{thm:fibrant-limits}
Let $\C$ be an essentially finite inverse category. 
Then, every Reedy fibrant $X : \C \to \UU$ has a fibrant limit.
\end{theorem}
\begin{proof}
By induction on the cardinality of $\C$. In the zero case, the limit is the unit type.

 Otherwise, let us consider the rank functor $\varphi : \C \to \oppo{{\strictN}}$.
 We choose an object $z : \C$ such that $\varphi_z$ is maximal; this is possible (constructively) since $\C$ is assumed to be essentially finite.
 Let us call $\C'$ the category that we get if we remove $z$ from $\C$; 
 that is, we set $\obj {\C'} \defeq \sm{x : \obj \C} \neg (x \steq z)$.
 Clearly, $\C'$ is still essentially finite and inverse.
 
 Let $X : \C \to \UU$ be Reedy fibrant. 
 We can write down the limit of $X$ 
 explicitly as
 \begin{equation}
  \sm{c : \prd{y : \obj \C} X_y} \prd{y,y' : \obj \C} \prd{f : \C(y,y')} Xf(c_y) \steq c_{y'}.
 \end{equation}
 Using that $z$ has no incoming non-identity arrows, this pretype is strictly isomorphic to
 \begin{equation} \label{eq:limit-2}
 \begin{alignedat}{1}
  \sm{c_z : X_z} &\sm{c : \prd{y : \obj {\C'}} X_y} \\
  &\phantom{\Sigma} \left(\prd{y : \obj {\C'}} \prd{f : \C(z,y)} Xf(c_z) \steq c_y\right) \times \left(\prd{y,y' : \obj {\C'}} \prd{f : \C(y,y')} Xf(c_y) \steq c_{y'}\right).
 \end{alignedat}
 \end{equation}

 Let us write $L$ for the limit of $X$ restricted to $\C'$,
 and let us further write $p$ for the canonical map $p : L \to M_z^X$. 
 Further, we write $q$ for the map $X_z \to M_z^X$.
 Then, \eqref{eq:limit-2} is strictly isomorphic to
 \begin{equation}\label{eq:limit-3}
  \sm{c_z : X_z} \sm{d : L} p(d) \steq q(c_z).
 \end{equation}
 This is the pullback of the span $L \xrightarrow{p} M_z^X \xleftarrow{q} X_z$.
 By Reedy fibrancy of $X$, the map $q$ is a fibration. 
 Thus, by Lemma~\ref{prop:fibrant-pullback}, the map from \eqref{eq:limit-3} to $L$ is a fibration.

 By the induction hypothesis, $L$ is essentially fibrant.
 This implies that \eqref{eq:limit-3} is essentially fibrant, as it is the domain of a fibration whose codomain is essentially fibrant.
\end{proof}

If $\C$ is an inverse category, we will denote by $\C^{<n}$ the full subcategory
of $\C$ consisting of all those objects of rank less than $n$.  Correspondingly,
for a given diagram $X$ over $\C$, we will denote by $X|n$ the restriction of
$X$ to $\C^{<n}$.

\newcommand{\tdeltop}[1]{\left(\op{\Delta_+}\right)^{<#1}}

\subsection{Fibrant Limits and Semi-Simplicial Types}

If $X$ is a Reedy fibrant diagram over $\C \defeq \tdeltop n$, we can restrict $X$ to $n
\sslash \C$, then take the limit of the corresponding functor.  With a slight
abuse of notation, we will denote such limit by $M_n^X$, even though $X$ is not
defined at $n$.

Note that a diagram $X$ over $\tdeltop{n+1}$ is Reedy fibrant if and only if its
restriction to $\tdeltop n$ is Reedy fibrant and the map $X_n \to M_n^X$ is a
fibration.  Hence, to give a Reedy fibrant diagram over $\tdeltop{n+1}$ is the
same as to give a Reedy fibrant diagram $X$ over $\tdeltop{n}$, together with a
fibration $Y$ over $M_n^X$.  We will refer to this extended diagram as $\langle
X, Y \rangle$.
By mutual induction on the natural number $n$, we can define a type $\sst_n$,
and a function $\ssk_n$ from $\sst_n$ to diagrams over $\tdeltop n$.  We start
with with $\sst_0 \defeq \unit$ and $\ssk_0(\ttt)$ set to the trivial diagram
over $\tdeltop 0$.
Then, we set
\begin{equation*}
\sst_{n+1} \defeq \sm{X : \sst_n} (M_{n}^{\ssk_{n} X} \to \UU)
\hspace*{1.2cm} \text{and} \hspace*{1.2cm}
\ssk_{n+1}(X, Y) \defeq \langle X, Y \rangle.
\end{equation*}
Above, we write $M_n^A$ to mean the type given by Theorem
\ref{thm:fibrant-limits} which is strictly isomorphic to the matching object of
$A$ at $n$ (which would otherwise only be a pretype).

This gives us a succinct alternative to the construction of Section
\ref{sec:semisimp}, where most of the hard work is encapsulated in the use of
Theorem \ref{thm:fibrant-limits}.

\section{Conclusions and Further Work} \label{sec:conclusions}

In the previous two sections, we have demonstrated how our 2-level theories can be used in two ways. 
First, our framework offers reasonable, easily justifiable ways of extending homotopy type theory.
Second, we can internalise results about homotopy type theory that, before, could only be stated meta-theoretically. 
In a suitable proof assistant which implements a 2-level theory, we could formalize many constructions that can at the moment only be done on paper. 
Our current article offers a demonstration of this possibility: we have shown
how some of the constructions about fibrant limits and diagrams can be
internalised.
From here, we could go into several directions.
We could, for example, internalise Shulman's result that diagrams over a model of type theory form again a model, preserving univalence~\cite{shulman_inversediagrams}.
Of course, for such an internalisation, we need to be careful to formulate all definitions and results constructively.

A more modest but (as we believe) worthwhile next goal is the construction of fibrant replacements.
With this, we can internalise the proof that any type carries the structure of an $\infty$-groupoid (a Kan semi-simplicial type), as it is given in~\cite[Remark and Corollary 16]{kraus_generaluniversalproperty}.
To do this, we would first define an $\infty$-groupoid to be a Reedy fibrant semi-simplicial type $X : \deltop \to \UU$ such that every fibration from $X_n$ to a horn is an equivalence (in the sense of homotopy type theory).
We can then, for a type $A : \UU$, consider the semi-simplicial type $\mathsf{Eq}_A$, defined to be the Reedy fibrant replacement of the functor that is constantly $A$.
It is shown in~\cite{kraus_generaluniversalproperty} that $\mathsf{Eq}_A$ is an $\infty$-groupoid in our sense, and the argument can easily be internalised. 
This construction is in fact not difficult and has in the current paper been omitted solely for reasons of space.

Our next significant project, supported by the 2-level theory, 
is the development of $(\infty,1)$-category theory.
By an $(\infty,1)$-category, which could also be called a \emph{Segal type}, we mean a Reedy fibrant semi-simplicial type $X$ for which the usual ``Segal maps'' $X_n \to X_1 \times_{X_0} \ldots \times_{X_0} X_1$ are equivalences.
It is likely that it is necessary to add degeneracies, and we expect that this can be done in the way presented by Harpaz~\cite{harpaz2015quasi}.

We believe that it is important to develop a theory of $(\infty,1)$-categories
type-theoretically, because the universe itself should be an
$(\infty,1)$-category; we expect that many infinite coherence problems become
approachable if we can set up some basic infrastructure, so that towers of
coherences could be formulated and handled in a clean way.

The most important application that we currently have in mind is the specification of \emph{higher inductive types} (HITs).
Although HITs are used frequently in the literature on homotopy type theory, we do not have a general syntactical specification yet.
The approach to define a general syntactical framework of HITs that is used in~\cite{gabe_HITs} seems to be promising, 
but suffers from the issue that an unmanageable number of coherences
needs to be handled manually.
We expect and hope that this can be resolved with the framework of $(\infty,1)$-categories that we plan to develop.

\bibliographystyle{alpha}
\bibliography{master}

\end{document}